\documentclass[a4paper,11pt]{article}
\usepackage{}
\usepackage[left=3.17cm,right=3.17cm,top=2.54cm,
headheight=0.5cm,headsep=0.54cm,bottom=2.54cm,footskip=0.79cm
]{geometry}

\usepackage[all]{xy}
\usepackage{amsmath,amsfonts,amssymb,amsthm,epsfig,amscd,comment,latexsym,psfrag}
\usepackage{nicefrac,xspace,tikz}

\usetikzlibrary{arrows}
\usetikzlibrary{decorations.markings}
\def\Z{\mathbb Z}

\def\C{\mathbb C}

\def\1{{\bf{1}}}

\usepackage[pdfstartview=FitH]{hyperref}

\def\footnoterule{\kern 1mm \hrule width 7cm \kern 2.2mm}%

 \newtheorem{thm}{Theorem}[section]
 \newtheorem{prp}[thm]{Proposition}

 \newtheorem{dfn}[thm]{Definition}
 
 \newtheorem{rmk}[thm]{Remark}

\begin{document}

\title{$\pi$-type Fermions and $\pi$-type KP hierarchy}

\author{
\  \ Na Wang\dag, Chuanzhong Li\ddag\footnote{Corresponding author:lichuanzhong@nbu.edu.cn.}\\
\dag\small Department of Mathematics and Statistics, Henan University, Kaifeng, 475001, China.\\
\ddag\small Department of Mathematics,  Ningbo University, Ningbo, 315211, China}

\date{}

\maketitle

\begin{abstract}
In this paper, we firstly construct $\pi$-type Fermions. According to these, we define $\pi$-type Boson-Fermion correspondence which is a generalization of the classical Boson-Fermion correspondence. We can obtain $\pi$-type symmetric functions $S_\lambda^\pi$ from the $\pi$-type Boson-Fermion correspondence, analogously to the way we get the Schur functions $S_\lambda$ from the classical Boson-Fermion correspondence (which is the same thing as the Jacobi-Trudi formula). Then as a generalization of KP hierarchy, we construct the $\pi$-type KP hierarchy and obtain its tau functions.

\end{abstract}
\noindent
{\bf 2010 Mathematics Subject Classification:}
 37K05,
37K10. 

{\bf Key words: }{Boson-Fermion Correspondence, $\pi$-type symmetric functions, $\pi$-type Fermions, $\pi$-type KP hierarchy.}

\section{Introduction}\label{sec1}
Two-dimensional Fermions and Boson-Fermion correspondence are well-known in mathematical physics. Meanwhile Young diagrams and symmetric functions are of interest to many researchers and have many applications in mathematics including combinatorics and representation theory of the symmetric and general linear group. There are many relations between them.

The Kakomtsev-Petviashvili (KP) hierarchy\cite{DKJM} is one of the most important integrable hierarchies and it arises in many different fields of mathematics and physics such as enumerative algebraic geometry, topological field and string theory. Schur functions have close relations with the tau functions of KP hierarchy. Schur functions in variables $x_1,x_2,\cdots,x_n$ are well known to give the characters of finite dimensional irreducible representations of the general linear groups $GL(n)$\cite{Mac,FH}. From \cite{MJD,NR,Jing}, Schur functions can be realized from vector operators and these vertex operators correspond to free Fermions acting on Bosonic Fock space. It turns out that the Boson-Fermion correspondence and the Jacobi-Trudi formula are the same thing, which tells us that Schur functions are solutions of differential equations in KP hierarchy, and the linear combinations of Schur functions with coefficients satisfied some relations (p\"uker relations) are also tau functions of KP hierarchy.

The $\pi$-type symmetric functions are upgraded from Schur functions in the same setting. The linear basis of $\pi$-type symmetric functions provides the structure of the universal character ring of group $H_\pi$ (subgroup of $GL(n)$)\cite{BJK,HW,Litt}. Like Schur functions, $\pi$-type symmetric functions can also be realized from vertex operators which are constructed in \cite{FJK}. Then free Fermions can be constructed and there exists for sure an integrable system. In this paper, we will construct this integrable system, and find that the $\pi$-type symmetric functions are the solutions of differential equations in this integrable system.

The paper is organized as follows. In section \ref{sect2}, we recall the $\pi$-type symmetric functions and vertex operators associated with them. In section \ref{sect3}, we recall Schur functions and KP hierarchy. In section \ref{sect4}, we define $\pi$-type Fermions and construct $\pi$-type Boson-Fermion correspondence from which we can calculate $\pi$-type symmetric functions. In section \ref{sect5}, we construct the $\pi$-type KP hierarchy and analyze its tau functions.
\section{$\pi$-type symmetric function and vertex operator}\label{sect2}
We begin this section with some notational preliminaries\cite{FJK}. Let $\Lambda({\bf x})$ be the ring of symmetric functions of a countably infinite alphabet of variables ${\bf x}=\{x_1,x_2,\cdots\}$. The power sum symmetric functions $p_n({\bf x})$ are
\[
p_n({\bf x})=\sum_{k}x_k^n.
\]
The operators $p_n$ and $n\partial_{p_n}$ ($\partial_{p_n}:=\partial/\partial_{p_n}$) give a representation of the infinite dimensional Heisenberg algebra generated by $a_n, n\in\Z, n\neq 0$ with the relation
\begin{equation}\label{hei}
[a_n,a_m]=n{\delta}_{n+m,0}.
\end{equation}
The vertex operator are defined with the help of Heisenberg algebras
\begin{eqnarray}
&&M(z,{\bf x})=\prod_k\frac{1}{1-zx_k}=\exp(\sum_{n=1}^\infty\frac{p_n}{n}z^n)\\
&&L(z,{\bf x})=\prod_k{(1-zx_k)}=\exp(-\sum_{n=1}^\infty\frac{p_n}{n}z^n)
\end{eqnarray}
where $L(z,{\bf x})=M(z,{\bf x})^{-1}$. In the special case $z=1$, we set $M({\bf x})=M(1,{\bf x})$ and $L({\bf x})=L(1,{\bf x})$. When ${\bf x}$ is to be understood, we often write $L(z,{\bf x})$ and $M(z,{\bf x})$ by $L(z)$ and $M(z)$ respectively for short.

For a Young diagram $\lambda$, let $\mathcal{T}^\lambda$ denote the set of semistandard tableaux $T$ of shape $\lambda$ with entries from $\{1,2,\cdots,n\}$, and let ${\bf x}^T=x_1^{\#1}x_2^{\#2}\cdots x_n^{\#n}$ where $\#k$ is the number of entries $k$ in $T$, then the Schur function
\[S_\lambda({\bf x})=\sum_{T\in\mathcal{T}^\lambda}{\bf x}^T
\]
Schur functions are an orthonormal basis of the ring $\Lambda({\bf x})$. The operation of symmetric function skew is defined by duality as
\[
\langle g^\perp f|h\rangle=\langle f|g\cdot f\rangle
\]
Given two Schur functions $S_\lambda$ and $S_\mu$, the skew Schur function $S_\mu^\perp S_\lambda=S_{\lambda/\mu}$.

The plethysm is defined as follows\cite{Mac}. Let $f({\bf x})=\sum_i y_i$. Consider these monomials as elements of a new countably infinite alphabet ${\bf y}=\{y_1,y_2,\cdots\}$. Then for any Schur function $S_\lambda$, the plethysm of $f$ by $S_\lambda$, $S_\lambda[f]({\bf x}):=S_\lambda ({\bf y})$ is the symmetric function of the composite alphabet. For any Young diagram $\pi$,
\begin{eqnarray}
&&M_\pi(z,{\bf x})=\prod_{T\in \mathcal {T}^\pi}\frac{1}{1-z{\bf x}^T}=\sum_{r\geq 0}z^rS_{(r)}[S_\pi]({\bf x})\label{ml1}\\
&&L_\pi(z,{\bf x})=\prod_{T\in \mathcal {T}^\pi}{(1-z{\bf x}^T)}=\sum_{r\geq 0}(-1)^rz^rS_{(1^r)}[S_\pi]({\bf x})
\end{eqnarray}
for arbitrary $\pi$, we have
\[
S_\lambda^\pi({\bf x})=L^\perp_\pi({\bf x})S_\lambda
\]
the symmetric functions of this type correspond to the branching rule from a module of the general linear group to (generically indecomposable) module of the $H_\pi$ subgroup.

Define
\begin{equation}
V_\pi(z)=M(z)L^\perp(z^{-1})\prod_{k>0}L^\perp_{\pi/(k)}(z^k),
\end{equation}
then we have\begin{equation}
S_\lambda^\pi=[{\bf z}^\lambda]V_\pi(z_1)V_\pi(z_2)\cdots V_\pi(z_k)\cdot 1
\end{equation}
where $[{\bf z}^\lambda]$ means selecting the coefficient of $z_1^{\lambda_1}z_2^{\lambda_2}\cdots z_k^{\lambda_k}$. In order to obtain the complete set of exchange relations between the $\pi$-type vertex operators it is necessary to introduce suitably constructed dual vertex operators $V^*_\pi(z)$.
\begin{thm}(theorem 1 in \cite{FJK})
For each partition $\pi$ and any $z$ let
\begin{eqnarray}
V_\pi(z)&:=&M(z)L^\perp(z^{-1})\prod_{k>0}L^\perp_{\pi/(k)}(z^k),\\
V_\pi^*(z)&:=&L(z)M^\perp(z^{-1})\prod_{k\geq 0}M^\perp_{\pi/(1^{2k+1})}(z^{2k+1})\prod_{k>0}L^\perp_{\pi/(k)}(z^{2k}),
\end{eqnarray}
where it is to be understood that all the Schur functions in $M(w),L(w),M^\perp(w)$ and $L^\perp(w)$, for any $w$, depend on the same sequence of variables $(x_1,x_2,\cdots)$ whose specification, for the sake of simplicity, has been suppressed.

Furthermore, let the associated full vertex operators $X^\pi(z)$ and $X^{*\pi}(z)$, constructed by adjoining zero mode contributions, be defined by
\begin{eqnarray}
X^\pi(z)&=&V_\pi(z)e^Kz^{H_0}:=\sum_{n\in\Z+1/2}X_j^\pi z^{-j-1/2+H_0}\\
X^{*\pi}(z)&=&V_\pi^*(z)e^{-K}z^{-H_0}:=\sum_{n\in\Z+1/2}X_j^{*\pi} z^{-j-1/2-H_0}
\end{eqnarray}
then we have the modes $X^\pi(z)$ and $X^{*\pi}(z)$ fulfil the free Fermion anticommutation relations of  a complex Clifford algebra:
\[
\{X_i^\pi, X_j^\pi\}=0,\ \{X_i^{*\pi}, X_j^{*\pi}\}=0,\ \{X_i^\pi, X_j^{*\pi}\}=\delta_{i+j,0}.
\]
where $\{\cdot,\cdot\}$ signifies an anticommutator.
\end{thm}
\section{Schur function, vertex operator and KP hierarchy}\label{sect3}
Let $\C[{\bf x}]=\C[x_1,x_2,\cdots]$ be the polynomial ring of infinitely many variables.
Although the number of variables is infinite, each polynomial itself is a finite sum of monomials, so involves only finitely many of the variables.
Bosons are operators  $\{a_n\}_{n\in{\Z},n\neq 0}$ satisfying relations (\ref{hei}).
The representation of Bosons on $\C[{\bf x}]$ is $a_n=\frac{\partial}{\partial x_n}, a_{-n}=nx_n$ for $n>0$. Denote $\frac{\partial}{\partial x_n}$ by $\partial_n$. Define
\begin{equation}\label{pqn}
\text{exp}(\sum_{m\geq 1}x_m k^m)=\sum_{n\geq 0} P_{(n)} k^n,\quad \text{exp}(\sum_{m\geq 1}\frac{\partial_m}{m} k^m)=\sum_{n\geq 0} Q_{(n)} k^n
\end{equation}
when $i<0$, we set $P_{(i)}=0,\ Q_{(i)}=0$. For any $m$, $P_{(m)}$ is a polynomial of variables $x_1,x_2,\cdots, x_m$, $P_{(m)}=P_{(m)}(x_1,x_2,\cdots,x_m)$.
In fact,
Replacing
$
x_n$ with the power sum symmetric function $p_n$,
we get
$
 P_{(n)} =S_{(n)}
$,
where $S_{(n)}$ is the Schur polynomial of Young diagram $(n)$. Let ${\bf x}=(x_1,x_2,\cdots)$ and $x_n=p_n(\tilde { x}_1,\tilde x_2,\cdots)$. Therefore,
for any Young diagram $\lambda=(\lambda_1,\lambda_2,\cdots,\lambda_l)$,
\begin{equation}\label{ps}
P_{\lambda}({\bf x})=S_{\lambda}(\tilde {\bf x})=\text{det}(h_{\lambda_i-i+j}(\tilde{\bf x}))_{1\leq i,j\leq l}
\end{equation}
where $h_{n}(\tilde{\bf x})$ is the $n$th complete symmetric function, i.e.,
\[
h_{n}(\tilde{\bf x})=\sum_{i_1\leq\cdots\leq i_n}\tilde x_{i_1}\cdots\tilde x_{i_n}.
\]

In the following,
we do not distinguish Young diagram $\lambda$, $P_\lambda$ and $S_\lambda$. The actions of $P_\lambda$ and $Q_\lambda$ on Young diagram $\mu$ are defined to be\cite{WWY2,W}
\begin{equation}
P_\lambda\cdot\mu:=\lambda\cdot\mu,\ \ \ Q_\lambda\cdot\mu:={\mu/\lambda}.
\end{equation}
where the multiplication $\lambda\cdot\mu$ satisfies the Littlewood-Richardson rule.

Introduce the vertex operators
\begin{eqnarray}
V^{\pm}(k)=\sum_{n\in\Z}V_n^{\pm}k^n=\exp(\pm\sum_{n=1}^\infty x_nk^n)\exp(\mp\sum_{n=1}^\infty \frac{\partial_n}{n}k^{-n}).
\end{eqnarray}
The operator $X_n^+$ is a raising operators of the Schur function, i.e.,
\begin{equation}
S_\lambda({\bf \tilde x})=P_\lambda({\bf x})=V_{\lambda_1}^+V_{\lambda_2}^+\cdots V_{\lambda_l}^+\cdot1
\end{equation}
for a partition $\lambda=(\lambda_1,\lambda_2,\cdots,\lambda_l)$.

By Boson-Fermion correspondence, there are three vector spaces which are isomorphic to each other, the polynomial ring $\C[{\bf x}]=\C[x_1,x_2,\cdots]$ of infinitely many variables ${\bf x}=(x_1,x_2,\cdots)$ which is called the Bosonic Fock space, the charge zero part of the Fermionic Fock space $\mathcal{F}$ which is the vector space based by the set of Maya diagrams, and the vector space $Y$ based by the set of Young diagrams. Therefore a Maya diagram $|u\rangle$ can be written as $$|u\rangle=|\lambda, n\rangle=|P_\lambda, n\rangle$$
where $n$ is the charge of $|u\rangle$. In special case, if the charge $n=0$, we also write the Maya diagram  $|u\rangle$ as $|\lambda\rangle$.

Let $f(z,{\bf x})$ be a function in space
$\C[z,z^{-1},x_1,x_2,\cdots].$
Define operators
\begin{equation}
e^K f(z,{\bf x}):=zf(z,{\bf x}),\ \ k^{H_0}f(z,{\bf x}):=f(kz,{\bf x}).
\end{equation}
Define the generating functions\cite{MJD}
\begin{eqnarray}
X(k)&=&\sum_{j\in \Z+\frac{1}{2}}X_{j}k^{-j-\frac{1}{2}}=V^{+}(k)e^K k^{H_0},\label{xx1}\\
X^*(k)&=&\sum_{j\in \Z+\frac{1}{2}}X^*_{j}k^{-j-\frac{1}{2}}=V^{-}(k)e^K k^{H_0}.\label{xx2}
\end{eqnarray}
It can be checked that\begin{equation}
\{X_i,X_j\}=0,
\{X_i^*, X_j^*\}=0,
\{X_i,X_j^*\}=\delta_{i+j,0}.
\end{equation}

\begin{dfn}
For an unknown function $\tau=\tau({\bf x})$,  the bilinear equation
\begin{equation}
\sum_{j\in{\Z+\frac{1}{2}}}X^*_j\tau\otimes X_{-j}\tau=0
\end{equation}
is called the KP hierarchy.
\end{dfn}
\section{$\pi$-type Boson-Fermion correspondence}\label{sect4}
We begin this section by recalling the definition of Maya diagram. Let an increasing sequence of half-integers\cite{MJD}
\[
{\bf u}=\{u_j\}_{j\geq 1},\ \text{with}\ u_1<u_2<u_3<\cdots,
\]
satisfy $u_{j+1}=u_j+1$ for all sufficiently large $j$.  Putting a black stone on the position $u_j$ for all $j$ and a white stone on every other half-integer position, we get a Maya diagram and denote it by $|{\bf u}\rangle$. Specially, the Maya diagram $|1/2,3/2,5/2,\cdots\rangle$ is denoted by $|\text{vac}\rangle$.

Fermions $\psi_j, \psi^*_j,\ j\in\frac{1}{2}+\Z$ are operators satisfying \[
\{\psi_i,\psi_j\}=0,
\{\psi_i^*, \psi_j^*\}=0,
\{\psi_i,\psi_j^*\}=\delta_{i+j,0}.
\]
The actions of Fermions $\psi_j, \psi^*_j$ on Maya diagrams are determined by
\begin{eqnarray}
\psi_j|{\bf u}\rangle &=&  \begin{cases}
 (-1)^{i-1}| \cdots,u_{i-1},u_{i+1},\cdots\rangle & \text{ if } u_i=-j \ \text{for some} \ i, \\
  0& \text{ otherwise},
\end{cases}  \\
\psi^*_j |{\bf u}\rangle &=&  \begin{cases}
 (-1)^{i}| \cdots,u_{i},j,u_{i+1},\cdots\rangle & \text{ if } u_i<j<u_{i+1} \ \text{for some} \ i, \\
  0& \text{ otherwise}.
\end{cases}
\end{eqnarray}
The generating functions of Fermions are
\[
\psi(k)=\sum_{j\in\Z+1/2}\psi_jk^{-j-1/2},\ \psi^*(k)=\sum_{j\in\Z+1/2}\psi^*_jk^{-j-1/2}.
\]
The normal order is defined as usual. Let
\begin{equation}\label{hhx}
H_n=\sum_{j\in\Z+1/2}:\psi_{-j}\psi_{j+n}^*: \ \ \text{and}\ \ \ H(x)=\sum_{n=1}^\infty x_nH_n.\end{equation}
For Maya diagrams $|{\bf u}\rangle$ and $|{\bf v}\rangle$, the pair $\langle{\bf v}|{\bf u}\rangle$ is defined by the formula
\[
\langle{\bf v}|{\bf u}\rangle=\delta_{v_1+u_1,0}\delta_{v_2+u_2,0}\cdots.
\]

The Boson-Fermion correspondence is the correspondence $$\Phi:\mathcal{F}\rightarrow \C[z,z^{-1},x_1,x_2,x_3,\cdots]$$
given by\[
|u\rangle\mapsto\Phi(|u\rangle):=\sum_{l\in\Z}\langle l|\exp(H(x))|u\rangle
\]
which is an isomorphism of vector spaces, where $|l\rangle$ is the Maya diagram obtained from $|\text{vac}\rangle$ by sliding the diagram bodily $l$ steps to the right (that is, $-l$ steps to the left if $l<0$), and the operator $H(x)$ is defined in (\ref{hhx}).

Under Boson-Fermion correspondence, the Fermions $\psi_j, \psi^*_j$, respectively, correspond to the operators $X_j, X_j^*$ defined in equations (\ref{xx1}) and (\ref{xx2}).

Let $\lambda$ be a Young diagram, and $\lambda'$ be its conjugate. The Frobenius notation $\lambda=(n_1,\cdots,n_l|m_1,\cdots,m_l)$ describes the Young diagram $\lambda$ by $n_i=\lambda_i-i$, $m_i=\lambda'_i-i$, where $l$ is the number of the boxes in the NW-SE diagonal line of $\lambda$.

The Boson-Fermion correspondence tells us that the basis vector
\[
\psi_{n_1}\cdots\psi_{n_l}\psi^*_{m_1}\cdots\psi^*_{m_l}|\text{vac}\rangle \ \  \text{for} \ n_1<\cdots<n_l<0\ \text{and}\ m_1<\cdots<m_l<0
\]
of Fermionic Fock space of charge zero goes over into the Schur function $S_\lambda$ multiplied by $(-1)^{\sum_{i=1}^l(m_i+{1}/{2})+l(l-1)/2}$, where $\lambda=(-n_1-1/2,\cdots,-n_l-1/2|-m_1-1/2,\cdots,-m_l-1/2)$. Under the correspondence between Young diagrams and Maya diagrams, we can also write $S_\lambda$ as
\begin{equation}\label{slam}
S_\lambda({\bf x})=\langle \text{vac}|e^{H(x)}|\lambda \rangle
\end{equation}

In the following, we will define $\pi$-type Boson-Fermion correspondence, from which we will get $\pi$-type symmetric functions.
Define
\begin{equation}
L^\perp (z)=\exp(-\sum_{n\geq 1}\frac{H_n}{n}z^n)
\end{equation}
which corresponds to $L^\perp (z)$ defined in Section 2 by Boson-Fermion correspondence, that is why we use the same notation.

According to the definition of operators $P_\lambda({\bf x})$ in equations (\ref{pqn}) and (\ref{ps}), The operator $P_\lambda(H_n)$ is defined by replacing $x_i$ in $P_\lambda({\bf x})$ with $\frac{1}{i}H_{in}$,  which is a plethysm in fact. For example, $P_{(2)}({\bf x})=\frac{1}{2}x_1^2+x_2$, then $P_{(2)}(H_n)=\frac{1}{2}H_{n}^2+\frac{1}{2}H_{2n}$.

Define \begin{equation}
L^\perp_\pi(z)=\exp(-\sum_{n=1}^\infty\frac{z^n}{n}P_\pi(H_n))
\end{equation}
where $\pi$ is a Young diagram. In special case, if $\pi=\emptyset $, the operator $L^\perp_\pi(z)=1$; if $\pi=(0)$, the operator $L^\perp_\pi(z)=1-z$ and if $\pi=(1)$, the operator $L^\perp_\pi(z)=L^\perp(z)$. When $z=1$, we denote $L^\perp_\pi(z)$ by $L^\perp_\pi$.

\begin{dfn}
Define $\pi$-type Fermions by
\begin{eqnarray}
\psi_j^\pi&:=&L^\perp_\pi\psi_j(L^\perp_\pi)^{-1}\\
\psi_j^{*\pi}&:=&L^\perp_\pi\psi_j^*(L^\perp_\pi)^{-1}
\end{eqnarray}
\end{dfn}
It is easy to check that the $\pi$-type Fermions satisfy
\[
\{\psi_i^\pi,\psi_j^\pi\}=0,
\{\psi_i^{*\pi}, \psi_j^{*\pi}\}=0,
\{\psi_i^\pi,\psi_j^{*\pi}\}=\delta_{i+j,0}.
\]
\begin{prp}
Under Boson-Fermion correspondence, $\pi$-type Fermions $\psi_j^\pi,
\psi_j^{*\pi}$ correspond to $X^\pi_j, X_j^{*\pi}$ defined in Theorem 2.1 respectively. Therefore, the conclusion in Theorem 2.1 holds naturally.
\end{prp}
\begin{proof}
Under Boson-Fermion correspondence, the operator $H_n, n>0$ corresponds to $\partial_n$, then $L^\perp_\pi(z)$ corresponds to
\[
\exp(-\sum_{n=1}^\infty\frac{z^n}{n}P_\pi(\partial_n))
\]
which is the same as $L^\perp_\pi(z)$ defined in Section 2. In appendix A of \cite{FJK}, they have proved that
\begin{eqnarray*}
L^\perp_\pi(z)M(w)&=&M(w)\prod_{k\geq 0}L^\perp_{\pi/(k)}(zw^K),\\
L^\perp_\pi(z)L(w)&=&L(w)\prod_{k\geq 0}L^\perp_{\pi/(1^{2k})}(zw^{2k})M^\perp_{\pi/(1^{2k+1})}(zw^{2k+1})
\end{eqnarray*}
where $$M(z)=\exp(\sum_{n=1}^\infty {x_n}z^n)$$ appeared in $X(z)$ and $L(z)=(M(z))^{-1}$ appeared in $X^*(z)$, and we know that Fermions $\psi_j$ and $\psi_j^*$ correspond to $X_j$ and $X^*_j$ under Boson-Fermion correspondence, respectively.
Then we obtain the conclusion.
\end{proof}

In the following, we will generalize the Boson-Fermion correspondence to $\pi$-type, from which we can calculate $\pi$-type symmetric functions. It turns out that the classical Boson-Fermion correspondence is the special case $\pi=\emptyset$ of the $\pi$-type Boson-Fermion correspondence.

\begin{dfn}
Let $\mathcal{F}$ denote the Fermionic Fock space based by the set of Maya diagrams, define
\begin{equation}
\Phi_\pi: \mathcal{F}\rightarrow\C[z,z^{-1},x_1,x_2,\cdots]
\end{equation}
by \begin{equation}
\Phi_\pi(|u\rangle)=\sum_{l\in\Z}z^l\langle l|e^{H(x)}L^\perp_\pi|u\rangle
\end{equation}
where $|u\rangle$ is a Maya diagram.
\end{dfn}

\begin{prp}
The correspondence $\Phi_\pi:\mathcal{F}\rightarrow\C[z, z^{-1},x_1,x_2,\cdots]$ defined above is an isomorphism of vector spaces.
\end{prp}

Under the correspondence between Maya diagrams and Young diagrams, we know that the charge zero Maya diagram  $|u\rangle$ corresponds to a Young diagram denoted by $\lambda$, and we denote $|u\rangle$ by $|\lambda\rangle$, then we get

\begin{prp}
For a Maya diagram $|\lambda\rangle$ which corresponds to the Young diagram $\lambda$, the $\pi$-type symmetric functions $S_\lambda^\pi({\bf x})$ can be obtained from
\begin{equation}
S_\lambda^\pi({\bf x})=\langle \text{vac}|e^{H(x)}L^\perp_\pi|\lambda\rangle
\end{equation}
\end{prp}

From the relations between $\psi_j^\pi,\psi_j^{*\pi}$ and $\psi_j,\psi_j^*$, we have
\begin{prp} If Young diagram $\lambda=(-n_1-1/2,\cdots,-n_l-1/2|-m_1-1/2,\cdots,-m_l-1/2)$ in the Frobenius notation, then $S_\lambda^\pi({\bf x})$ can be obtained from
\[
\langle \text{vac}|e^{H(x)}\psi^\pi_{n_1}\cdots\psi^\pi_{n_l}\psi^{*\pi}_{m_1}\cdots\psi^{*\pi}_{m_l}|\text{vac}\rangle \ \  \text{for} \ n_1<\cdots<n_l<0\ \text{and}\ m_1<\cdots<m_l<0
\]
 by multiplying $(-1)^{\sum_{i=1}^l(m_i+{1}/{2})+l(l-1)/2}$.
\end{prp}
Take an example, we will calculate $S_{\begin{tikzpicture}
\draw [step=0.15](0,0) grid(0.3,0.15);
\end{tikzpicture}}^{(2)}$ in two ways. We will need the actions of $H_n$ on Maya diagrams.
From the actions of Fermions on Maya diagrams, we get
the action of $H_1$ on Maya diagram, that is $H_1$ sending a Maya diagram $|{\bf{u}}\rangle$ to
the sum over all Maya diagrams who can be obtained from $|{\bf{u}}\rangle$ by moving a black stone  to the right.
Define
\begin{equation}\label{hh}
\text{exp}(\sum_{m\geq 1}\frac{H_m}{m} k^m)=\sum_{n\geq 0} Q_{(n)} k^n,\ \ \text{exp}(\sum_{m\geq 1}\frac{H_{-m}}{m} k^m)=\sum_{n\geq 0} P_{(n)} k^n
\end{equation}
The action of $Q_{(m)}$ on Maya diagram is that $Q_{(m)}$ sending a Maya diagram $|{\bf{u}}\rangle$ to
the sum over all Maya diagrams who can be obtained from $|{\bf{u}}\rangle$ by moving black stones $n$ times to the right and no one black stone is moved twice. and $Q_{(1^m)}$ sends a Maya diagram $|{\bf{u}}\rangle$ to
the sum over all Maya diagrams who can be obtained from $|{\bf{u}}\rangle$ by moving black stones $n$ times to the right and no two adjacent black stones move at the same time\cite{Wang,W}. The actions of $P_\lambda$ on Maya diagram is similar to that of $Q_\lambda$ on Maya diagram but shifting the black stones to the left. Then
the action of $H_m$ on a Maya diagram can be obtained from the actions of $P_n, Q_n$ on this Maya diagram, and we get that the action of $H_m$ on Maya diagram is moving black stones $|m|$ times, to the right if $m>0$ and to the left if $m<0$. Then we can calculate  $S_{\begin{tikzpicture}
\draw [step=0.15](0,0) grid(0.3,0.15);
\end{tikzpicture}}^{(2)}$.
The first way, since
\begin{eqnarray*}
\psi^{(2)}_{-3/2}&=&\psi_{-3/2}-\psi_{1/2}-\psi_{-1/2}Q+\psi_{3/2}Q+\cdots\\
\psi^{*(2)}_{-1/2}&=&\psi^*_{-1/2}+\psi^*_{1/2}Q+\psi^*_{3/2}Q_{2}+\cdots
\end{eqnarray*}
then
\begin{eqnarray*}
S_{\begin{tikzpicture}
\draw [step=0.15](0,0) grid(0.3,0.15);
\end{tikzpicture}}^{(2)} &=&\langle\text{vac}|e^{H(x)}\psi_{-3/2}^{(2)}\psi^{*(2)}_{-1/2}|\text{vac}\rangle\\
&=&\langle\text{vac}|e^{H(x)}(\psi_{-3/2}-\psi_{1/2})\psi^*_{-1/2}|\text{vac}\rangle\\
&=& S_{\begin{tikzpicture}
\draw [step=0.15](0,0) grid(0.3,0.15);
\end{tikzpicture}}-S_0=\frac{1}{2}x_1^2+x_2-1.
\end{eqnarray*}
and in the second way, we know that for Maya diagram
\[
\begin{tikzpicture}
\draw (-3.4,-0.1) node{$|\gamma\rangle=$};
\fill (0,0) circle(0.1);
\draw (0,-0.5) node{$\frac{1}{2}$};
\draw (0.5,0) circle(0.1);
\draw (0.5,-0.5) node{$\frac{3}{2}$};
\fill (1,0) circle(0.1);
\draw (1,-0.5) node{$\frac{5}{2}$};
\fill (1.5,0) circle(0.1);
\draw (1.5,-0.5) node{$\frac{7}{2}$};
\fill (2,0) circle(0.1);
\draw (2,-0.5) node{$\frac{9}{2}$};
\draw (2.5,0) node{$\cdots$};
\fill (-0.5,0) circle(0.1);
\draw (-0.6,-0.5) node{$-\frac{1}{2}$};
\draw (-1,0) circle(0.1);
\draw (-1.1,-0.5) node{$-\frac{3}{2}$};
\draw (-1.5,0) circle(0.1);
\draw (-1.6,-0.5) node{$-\frac{5}{2}$};
\draw (-2,0) circle(0.1);
\draw (-2.1,-0.5) node{$-\frac{7}{2}$};
\draw (-2.5,0) node{$\cdots$};

\end{tikzpicture}
\]
we have that $H_m\gamma=0$ when $m>2$. Then
\begin{eqnarray*}
S_{\begin{tikzpicture}
\draw [step=0.15](0,0) grid(0.3,0.15);
\end{tikzpicture}}^{(2)}&=&\langle\text{vac}|e^{H(x)}e^{-\sum_{n=1}^\infty\frac{1}{2n}(H_n^2+H_{2n})}\psi_{-3/2}\psi_{-1/2}|\text{vac}\rangle\\
&=&\langle\text{vac}|e^{H(x)}(1-\frac{1}{2}(H_1^2+H_2))|\gamma\rangle\\
&=&\langle\text{vac}|e^{H(x)}(1-Q_2)|\gamma\rangle\\
&=& S_{\begin{tikzpicture}
\draw [step=0.15](0,0) grid(0.3,0.15);
\end{tikzpicture}}-S_0=\frac{1}{2}x_1^2+x_2-1.
\end{eqnarray*}
\section{$\pi$-type KP hierarchy}\label{sect5}
In this section, we will define the $\pi$-type KP hierarchy and discuss its tau functions.
\begin{dfn}
For an unknown charge zero state $|u\rangle$ in $\mathcal{F}$, the bilinear equation
\begin{equation}\label{pikp0}
\sum_{j\in\Z+\frac{1}{2}}\psi^{*\pi}_j|u\rangle\otimes\psi^{\pi}_{-j}|u\rangle=0
\end{equation}
is called the $\pi$-type KP hierarchy.
\end{dfn}
Under Boson-Fermion correspondence, this definition can be written into
\begin{dfn}
For an unknown function $\tau=\tau({\bf x})$, the bilinear equation
\begin{equation}\label{pikp}
\sum_{j\in\Z+\frac{1}{2}}X^{*\pi}_j\tau\otimes X^{\pi}_{-j}\tau=0
\end{equation}
is called the $\pi$-type KP hierarchy.

\end{dfn}

We write $V_\pi(z)=\sum_{n\in \Z}V^\pi_nz^n$ and $V_\pi^*(z)=\sum_{n\in\Z}V^{*\pi}_nz^n$. It can be check that $V^\pi_n$ and $V^{*\pi}_m$ satisfy
$V^\pi_nV^\pi_m+V^\pi_{m-1}V^\pi_{n+1}=0$, $V^{*\pi}_nV^{*\pi}_m+V^{*\pi}_{m-1}V^{*\pi}_{n+1}=0$ and $V^\pi_nV^{*\pi}_m+V^{*\pi}_{m+1}V^\pi_{n-1}=\delta_{n+m,0}$. From the relations between $V^\pi_n, V^{*\pi}_n$ and $X^\pi_n, X^{*\pi}_n$, the equation (\ref{pikp}) can be rewritten into
\begin{equation}\label{pikp1}
\sum_{n+m=-1}V^{*\pi}_n\tau\otimes V^{\pi}_m\tau=0
\end{equation}

Suppose $\tau=\tau({\bf x})$ is a solution of $\pi$-type KP hierarchy (\ref{pikp1}), then $V^\pi(\alpha)\tau$ solves (\ref{pikp1}) again with an arbitrary constant $\alpha\in\C^\times$.

In the following, we will discuss the differential equations in the $\pi$-type KP hierarchy and their solutions. From the relations between Fermions and $\pi$-type Fermions, the equation (\ref{pikp0}) can be rewritten into
\[
\sum_{j\in\Z+\frac{1}{2}} L^\perp_\pi\psi_j^*(L^\perp_\pi)^{-1}\tau\otimes L^\perp_\pi\psi_{-j}(L^\perp_\pi)^{-1}\tau=0
\]
Multiplied by $(L^\perp_\pi)^{-1}\otimes (L^\perp_\pi)^{-1}$, the equation above turns into
\begin{equation}\label{pikp2}
\sum_{j\in\Z+\frac{1}{2}}\psi^*_j M^\perp_\pi\tau\otimes \psi_{-j}M^\perp_\pi\tau=0
\end{equation}
From this, we can obtain
\begin{prp}\label{relation}
 If $\tau$ is a solution of the $\pi$-type KP hierarchy, then $M^\perp_\pi\tau$ is a solution of KP hierarchy. If $\tau$ is a solution of KP hierarchy, then $L^\perp_\pi\tau$ is a solution of $\pi$-type KP hierarchy.
 \end{prp}

 From (\ref{ml1}), we have
 \[
 M^\perp_\pi=\sum_{n\geq 0}(S_{(n)}[S_\pi]({\bf x}))^\perp
 \]
 From the definition of plethysm, the relation $S_{(n)}[S_\pi]=\sum_{\lambda}a_{(n)\pi}^\lambda S_\lambda$ holds for $|\lambda|=n\cdot|\pi|$, and it has been proved that $a_{(n)\pi}^\lambda$ are nonnegative integers\cite{Mac}. Then
 \[
 M^\perp_\pi=\sum_{n\geq 0}\sum_\lambda a_{(n)\pi}^\lambda (S_\lambda)^\perp=\sum_{n\geq 0}\sum_\lambda a_{(n)\pi}^\lambda Q_\lambda
 \]
 The action of $P_\lambda$ on Maya diagram is clearly known\cite{Wang}. Let $\lambda=(\lambda_1,\lambda_2,\cdots,\lambda_k)$ be a Young diagram. The action of $P_\lambda$ on Maya diagram $|{\bf{u}}\rangle$ includes $k$ steps corresponding to $k$ in $\lambda$. The first step is $P_{(\lambda_1)}$ acting on $|{\bf{u}}\rangle$ which we have introduced above, and the position, where the black stone is moved, is labeled by $1$; The second step is $P_{(\lambda_2)}$ acting on all the Maya diagrams obtained from $P_{(\lambda_1)}\cdot|{\bf{u}}\rangle$ and the position, where the black stone is moved, is labeled by $2$; Continuing until the $k$th step, the operator $P_{(\lambda_k)}$ acts on all the Maya diagrams obtained from $P_{(\lambda_{k-1})}\cdots P_{(\lambda_2)}P_{(\lambda_1)}\cdot|{\bf{u}}\rangle$, and the position, where the black stone is moved, is labeled by $k$. We define $P_\lambda$ sending Maya diagram $|{\bf{u}}\rangle$ to the sum over all Maya diagrams obtained from $k$ steps above and satisfied the following situation: from right to left, one looks at the first $l$ entries in the list (for any $l$ between $1$ and $\lambda_1+\lambda_2+\cdots+\lambda_k$), each integer $p$ between $1$ and $k-1$ occurs at least as many times as the next integer $p+1$.

 Choose $\tau$ in equation (\ref{pikp2}) in the form of linear combination of all charge zero Maya diagrams $\tau=\sum_{|u\rangle}c(|u\rangle)|u\rangle$, where the coefficient $c(|u\rangle)\in\C$. Let $|\mu\rangle$ be a Maya diagram of charge $1$ and $|\nu\rangle$ of charge $-1$, the coefficient of $|\mu\rangle\otimes |\nu\rangle$ in equation (\ref{pikp2}) is zero. From this, we will get many differential equations whose solutions include $\pi$-type symmetric functions.
\begin{prp} In $\pi$-type KP hierarchy, the tau function $\tau$ is a solution if and only if the coefficients $c(|u\rangle)$ in $\tau=\sum_{|u\rangle}c(|u\rangle)|u\rangle$ satisfy Pl\"uker relations, i.e.,
the following equation holds for any $|\mu\rangle$ and $|\nu\rangle$ whose charges are $1$ and $-1$ respectively
 \begin{equation}\label{pieq1}
\sum_{j}(-1)^j\sum_{(n),\lambda}\sum_{(m),\lambda'}a_{(n)\pi}^\lambda a_{(m)\pi}^{\lambda'}c(P_\lambda|\mu-\mu_j\rangle)c(P_{\lambda'}|\nu+\mu_j\rangle)=0
\end{equation}
here, Maya diagrams are signed Maya diagrams whose definition can be found in \cite{MJD}.
\end{prp}

 When $\pi=\emptyset$, the equation (\ref{pieq1}) turns into
 \begin{equation}\label{piempt}
\sum_{j}(-1)^jc(|\mu-\mu_j\rangle)c(|\nu+\mu_j\rangle)=0
\end{equation}
which is the Pl\"uker relations of KP hierarchy (the equation (10.3)  in \cite{MJD}).

 When $\pi=(1)=\begin{tikzpicture}
\draw [step=0.2](0,0) grid(0.2,0.2);
\end{tikzpicture}$, the equation (\ref{pieq1}) turns into
 \begin{equation}
\sum_{j}(-1)^j\sum_{n,m\geq 0}c(P_n|\mu-\mu_j\rangle)c(P_m|\nu+\mu_j\rangle)=0
\end{equation}
which in fact is the same as (\ref{piempt}).

When $\pi=(2)=\begin{tikzpicture}
\draw [step=0.2](0,0) grid(0.4,0.2);
\end{tikzpicture}$, since $a_{(n)(2)}^\lambda =1$ if and only if $\lambda$ is an even partition ($\lambda_i$ are even numbers) of $2n$, otherwise $a_{(n)(2)}^\lambda =0$, then the equation (\ref{pieq1}) turns into
\begin{equation}\label{pi22}
\sum_{j}(-1)^j\sum_{\lambda,\lambda' \text{even}}c(P_\lambda|\mu-\mu_j\rangle)c(P_{\lambda'}|\nu+\mu_j\rangle)=0.
\end{equation}
 In special case, let \begin{equation}
\begin{tikzpicture}
\draw (-3.4,-0.1) node{$\mu=$};
\fill (0,0) circle(0.1);
\draw (0,-0.5) node{$\frac{1}{2}$};
\fill (0.5,0) circle(0.1);
\draw (0.5,-0.5) node{$\frac{3}{2}$};
\fill (1,0) circle(0.1);
\draw (1,-0.5) node{$\frac{5}{2}$};
\fill (1.5,0) circle(0.1);
\draw (1.5,-0.5) node{$\frac{7}{2}$};
\fill (2,0) circle(0.1);
\draw (2,-0.5) node{$\frac{9}{2}$};
\draw (2.5,0) node{$\cdots$};
\fill (-0.5,0) circle(0.1);
\draw (-0.6,-0.5) node{$-\frac{1}{2}$};
\draw (-1,0) circle(0.1);
\draw (-1.1,-0.5) node{$-\frac{3}{2}$};
\draw (-1.5,0) circle(0.1);
\draw (-1.6,-0.5) node{$-\frac{5}{2}$};
\draw (-2,0) circle(0.1);
\draw (-2.1,-0.5) node{$-\frac{7}{2}$};
\draw (-2.5,0) node{$\cdots$};

\end{tikzpicture}
\end{equation}
and\begin{equation}
\begin{tikzpicture}
\draw (-3.4,-0.1) node{$\nu=$};
\draw (0,0) circle(0.1);
\draw (0,-0.5) node{$\frac{1}{2}$};
\draw (0.5,0) circle(0.1);
\draw (0.5,-0.5) node{$\frac{3}{2}$};
\fill (1,0) circle(0.1);
\draw (1,-0.5) node{$\frac{5}{2}$};
\fill (1.5,0) circle(0.1);
\draw (1.5,-0.5) node{$\frac{7}{2}$};
\fill (2,0) circle(0.1);
\draw (2,-0.5) node{$\frac{9}{2}$};
\draw (2.5,0) node{$\cdots$};
\draw (-0.5,0) circle(0.1);
\draw (-0.6,-0.5) node{$-\frac{1}{2}$};
\fill (-1,0) circle(0.1);
\draw (-1.1,-0.5) node{$-\frac{3}{2}$};
\draw (-1.5,0) circle(0.1);
\draw (-1.6,-0.5) node{$-\frac{5}{2}$};
\draw (-2,0) circle(0.1);
\draw (-2.1,-0.5) node{$-\frac{7}{2}$};
\draw (-2.5,0) node{$\cdots$};

\end{tikzpicture}
\end{equation}
the equation (\ref{pi22}) equals
\begin{equation}
\sum_{\lambda,\lambda' \text{even}}(c(\lambda)c(\lambda'\cdot\begin{tikzpicture}
\draw [step=0.15](0,0) grid(0.3,0.3);
\end{tikzpicture})-c(\lambda\cdot\begin{tikzpicture}
\draw [step=0.15](0,0) grid(0.15,0.15);
\end{tikzpicture})c(\lambda'\cdot\begin{tikzpicture}
\draw [step=0.15](0,0) grid(0.3,0.15);
\draw [step=0.15](0,-0.15) grid(0.15,0);
\end{tikzpicture})+c(\lambda\cdot\begin{tikzpicture}
\draw [step=0.15](0,0) grid(0.3,0.15);
\end{tikzpicture})c(\lambda'\cdot\begin{tikzpicture}
\draw [step=0.15](0,0) grid(0.15,0.3);
\end{tikzpicture}))=0
\end{equation}
where $\lambda\cdot\mu=\sum_{\nu}N_{\lambda\mu}^\nu\nu,\ N_{\lambda\mu}^\nu\in\Z_{\geq 0}$ satisfies the Littlewood-Richardson rule. Replacing $c(\lambda)$ by $Q_\lambda\tau$, we get the differential equation
\begin{equation}
\sum_{\lambda,\lambda' \text{even}}(Q_{\lambda}\tau\cdot Q_{\lambda'\cdot\begin{tikzpicture}
\draw [step=0.15](0,0) grid(0.3,0.3);
\end{tikzpicture}}\tau-Q_{\lambda\cdot\begin{tikzpicture}
\draw [step=0.15](0,0) grid(0.15,0.15);
\end{tikzpicture}}\tau\cdot Q_{\lambda'\cdot\begin{tikzpicture}
\draw [step=0.15](0,0) grid(0.3,0.15);
\draw [step=0.15](0,-0.15) grid(0.15,0);
\end{tikzpicture}}\tau+Q_{\lambda\cdot\begin{tikzpicture}
\draw [step=0.15](0,0) grid(0.3,0.15);
\end{tikzpicture}}\tau\cdot Q_{\lambda'\cdot\begin{tikzpicture}
\draw [step=0.15](0,0) grid(0.15,0.3);
\end{tikzpicture}}\tau)=0
\end{equation}
We can similarly write the differential equations in $\pi$-type KP hierarchy, i.e., replacing $c(\lambda)$ by $Q_\lambda\tau$ in (\ref{pieq1}), therefore,
\begin{prp}
The differential equations in the $\pi$-type KP hierarchy are
 \begin{equation}
\sum_{j}(-1)^j\sum_{(n),\lambda}\sum_{(m),\lambda'}a_{(n)\pi}^\lambda a_{(m)\pi}^{\lambda'}Q_{P_\lambda|\mu-\mu_j\rangle}\tau\cdot Q_{P_{\lambda'}|\nu+\mu_j\rangle}\tau=0
\end{equation}
where $|\mu\rangle$ and $|\nu\rangle$  are two Maya diagrams whose charges are $1$ and $-1$ respectively. Choose $|\mu\rangle$ and $|\nu\rangle$ as before, we get
 \begin{equation}\label{pieq11}
\sum_{(n),\lambda}\sum_{(m),\lambda'}a_{(n)\pi}^\lambda a_{(m)\pi}^{\lambda'}(Q_{\lambda}\tau\cdot Q_{\lambda'\cdot\begin{tikzpicture}
\draw [step=0.15](0,0) grid(0.3,0.3);
\end{tikzpicture}}\tau-Q_{\lambda\cdot\begin{tikzpicture}
\draw [step=0.15](0,0) grid(0.15,0.15);
\end{tikzpicture}}\tau\cdot Q_{\lambda'\cdot\begin{tikzpicture}
\draw [step=0.15](0,0) grid(0.3,0.15);
\draw [step=0.15](0,-0.15) grid(0.15,0);
\end{tikzpicture}}\tau+Q_{\lambda\cdot\begin{tikzpicture}
\draw [step=0.15](0,0) grid(0.3,0.15);
\end{tikzpicture}}\tau\cdot Q_{\lambda'\cdot\begin{tikzpicture}
\draw [step=0.15](0,0) grid(0.15,0.3);
\end{tikzpicture}}\tau)=0
\end{equation}
\end{prp}
The solutions of these equations is known from the discussions before. When $\pi=\emptyset$, the equation (\ref{pieq11}) turns into
\[
\tau\cdot Q_{\begin{tikzpicture}
\draw [step=0.15](0,0) grid(0.3,0.3);
\end{tikzpicture}}\tau-Q_{\begin{tikzpicture}
\draw [step=0.15](0,0) grid(0.15,0.15);
\end{tikzpicture}}\tau\cdot Q_{\begin{tikzpicture}
\draw [step=0.15](0,0) grid(0.3,0.15);
\draw [step=0.15](0,-0.15) grid(0.15,0);
\end{tikzpicture}}\tau+Q_{\begin{tikzpicture}
\draw [step=0.15](0,0) grid(0.3,0.15);
\end{tikzpicture}}\tau\cdot Q_{\begin{tikzpicture}
\draw [step=0.15](0,0) grid(0.15,0.3);
\end{tikzpicture}}\tau=0
\]
which is the KP equation
\[
\frac{3}{4}\frac{\partial^2u}{\partial x_2^2}=\frac{\partial}{\partial x}\left(\frac{\partial u}{\partial x_3}-\frac{3}{2}u\frac{\partial u}{\partial x}-\frac{1}{4}\frac{\partial^3u}{\partial x^3}\right).
\]

Then we have the following remark.
\begin{rmk}
The $\pi$-type KP hierarchy is quite different from other types of classical KP systems from the point of different types of Lie algbras (like BKP, CKP and so on \cite{DKJM,MJD}).
The relation between the tau functions of BKP, CKP and so on and the tau function of the classical KP hierarchy is more complicated than the relation between the tau function of the $\pi$-type KP hierarchy and the tau function of the classical KP hierarchy as mentioned in the Proposition \ref{relation}.
\end{rmk}

\section*{Acknowledgements}
The authors gratefully acknowledge the support of Professors Ke Wu, Zi-Feng Yang, Shi-Kun Wang.
Chuanzhong Li is supported by the National Natural Science Foundation
of China under Grant No. 11571192 and K. C. Wong Magna Fund in Ningbo University.

\end{document}